\documentclass[11pt]{article}

\usepackage{amsmath,amssymb,amsfonts}
\usepackage[lined]{algorithm2e}
\usepackage{times}
\renewcommand{\baselinestretch}{1.0}
\def\inline#1:{\par\vskip 7pt\noindent{\bf #1:}\hskip 10pt}

\newtheorem{theorem}{Theorem}[section]

\newtheorem{assert}{Assertion}

\newtheorem{corollary}[theorem]{Corollary}

\newcommand{\qed}{\hfill $\Box$ \medbreak}
\newenvironment{proof}{\noindent {\bf Proof.}}{\qed}

\newcommand{\uniform}{\mbox{\rm uniform}}
\newcommand{\kk}{\mbox{\rm k}}

\def\cT{{\cal T}}
\def\cA{{\cal A}}

\newenvironment{smallitemize} {
  \begin{list}{$\bullet$} {\setlength{\parsep}{0pt}
\setlength{\itemsep}{0pt}} } { \end{list} }

\long\def\jump#1\finjump{}

\title{Collaborative Search on the Plane without Communication}
\author{
Ofer Feinerman\thanks{
The Louis and Ida Rich Career Development Chair, The Weizmann Institute of Science, Rehovot, Israel.
E-mail: {\tt feinermanofer@gmail.com}.
Supported by the Israel Science Foundation (grant 1694/10).}
\and
Amos Korman\thanks{
CNRS and University Paris Diderot, Paris, France.
E-mail: {\tt \,amos.korman@liafa.jussieu.fr}.}
\and Zvi Lotker \thanks{Ben-Gurion University of the Negev, Beer-Sheva, Israel.
E-mail: {\tt \,zvilo@bgu.ac.il}.}
\and  Jean-S\'ebastien Sereni\thanks{CNRS and University Paris Diderot, Paris, France
   and  Charles University, Prague, Czech Republic.
 E-mail: {\tt \,sereni@kam.mff.cuni.cz}}
}

\date{}

\begin{document}
\begin{titlepage}
\def\thepage{}
\maketitle

\begin{abstract}
We generalize the classical cow-path problem \cite{BCR91,DFG2006,KRT96,KSW86} into a question that is relevant for collective foraging in animal groups.
Specifically, we consider a setting in which $k$ identical (probabilistic) agents, initially placed at some central location,
collectively search for a treasure in the two-dimensional plane. The treasure is placed at a target location
by an adversary and the goal is to find it as fast as possible as a function of both $k$ and $D$, where $D$ is the distance
between the central location and the target. This is biologically motivated by cooperative, central place foraging such as performed by ants around their nest. In this type of search there is a strong preference to locate nearby food sources before those that are further away.
Our focus is on trying to find what can be achieved if  communication is limited or altogether absent. Indeed, to avoid
overlaps agents must be highly dispersed making communication difficult. Furthermore, if agents do not commence the search  in synchrony then even initial communication is problematic. This holds, in particular, with respect to the question of whether the agents can communicate and conclude their total number, $k$. It turns out that the  knowledge of $k$ by the individual agents is crucial for performance. Indeed, it is a straightforward observation that the time required for finding the treasure is
$\Omega(D+ D^2/k)$, and we show in this paper that this bound can be matched if the agents have knowledge of $k$ up to some constant approximation.

We present an almost tight bound for the competitive penalty that must be paid, in the running time, if agents have no information about $k$. Specifically, on the negative side, we show that   in such a case, there is no algorithm whose competitiveness is $O(\log k)$. On the other hand, we show that for every constant $\epsilon> 0$, there exists a rather simple uniform search algorithm which is $O( \log^{1+\epsilon} k)$-competitive. 
In addition, we give a lower bound for the setting in which agents are given some estimation of $k$. As a special case, this lower bound implies that for any constant $\epsilon>0$, if each agent  is  given a (one-sided)  $k^\epsilon$-approximation to $k$, then the competitiveness is 
$\Omega(\log k)$.
Informally, our results imply that the agents can potentially perform well without any knowledge of their total number $k$, however, to further improve, they must be given a relatively good approximation of $k$. 
Finally, we propose a uniform algorithm that is both efficient and extremely simple  suggesting its relevance for actual biological scenarios.

\end{abstract}

\end{titlepage}

%%%%%%%%%%%%%%%%%%%%%%%
\section{Introduction}\label{sec:introduction}
\paragraph{Background and Motivation:}
The universality of search behaviors is reflected in multitudes of studies in different fields including control systems, distributed computing and biology. We use tools from distributed computing to study a biologically inspired scenario in which a group of agents, initially located at one central location, cooperatively search for treasures in the plane. The goal of the search is to locate nearby treasures as fast as possible and at a rate that scales well with the number of participating agents.

A variety of animal species search for food around a central location that serves the search's initial point, final destination or both \cite{OP79}. This central location could be a food storage area, a nest where offspring are reared or simply a sheltered or familiar environment. Central place foraging holds a strong preference to locating nearby food sources before those that are further away. Possible reasons for that are, for example: (1) decreasing predation risk \cite{K80}, (2) increasing the rate of food collection once a large quantity of food is found \cite{OP79}, (3) holding a territory without the need to reclaim it  \cite{GKDL94,K80,MR79}, and (4) the ease of navigating back after collecting the food using familiar landmarks \cite{CDGW92}.

Searching in groups can increase foraging efficiency  \cite{HW90}. In some extreme cases, food is so scarce that group searching is believed to be required for survival  \cite{CCG99,JS98}. Proximity of the food source to the central location is again important in this case. For example, in the case of recruitment, a nearby food source would be beneficial not only to the individual that located the source but also increase the subsequential retrieval rate for many other collaborators  \cite{T77}. Foraging in groups can also facilitate the defense of larger territories~\cite{S71}. Eusocial insects (e.g., bees and ants) engage in highly cooperative foraging, this can be expected as these insects reduce competition between individuals to a minimum and share any food that is found. Social insects often live in a single nest or hive which naturally makes their foraging patterns central.

Little is known about the communication between the foragers, but it is believed that in some scenarios
  communication may become impractical~\cite{HM85}. This holds, for example, if the foragers start the search at different times and remain far apart (which may be necessary to avoid unnecessary overlaps). Hence, the question of how efficient can the search be if the communication is limited, or altogether absent, is of great importance.
%It is also important to note that decreasing the distance at which a food source is found is not identical to decreasing the time to find it: the beginning to end distance is shorter on a convoluted track than on a straight track of the same length.

%The question of how efficient can the search be if the
In this paper, we theoretically address general questions of collective searches in the particular natural setting described above.
More precisely, our setting consists of  $k$ identical (probabilistic) agents, initially placed at some central location,
which collectively search for a treasure in the two-dimensional plane. The treasure is placed by an adversary at some target location at a distance $D$ from the central location, where $D$ is unknown to the agents.
 The goal of the agents is to find the treasure as fast as possible, where the time complexity is evaluated as a function of both $k$ and $D$.

In the context of search algorithms, evaluating the time as a function of $D$ was first introduced in the classical paper \cite{BCR91} by Baeza-Yates et al., which studied the {\em cow-path} problem (studied also in \cite{DFG2006,KRT96,KSW86}). Our setting generalizes the one used for the cow-path problem as we consider multiple identical agents instead of a single agent  (a cow in their terminology).
Indeed,  in this distributed setting, we are concerned with the
%(significantly) reducing the running time, as more agents participate in the search. This is reflected by the important notion of
 {\em speed-up} measure (see also, \cite{AB96, AB97,FDPS11,FGKP04}), which aims to capture  the impact of using $k$ searchers in comparison to using a single one. Note that the objectives of quickly finding nearby treasures and having significant speed-up may be at conflict. That is, in order to ensure that nearby treasures are quickly found, a large enough fraction of the search force must be deployed near the central location. In turn, this crowding can potentially lead to overlapping searches that decrease individual efficiency.

It is a rather straightforward observation that the time required for finding the treasure is
$\Omega(D+ D^2/k)$.
Our focus is on the question of how  agents can approach this bound if their communication  is limited or even completely absent. 
In particular, as information of foraging group size may not be available to the individual searchers, we concentrate our attention on the question of how important it is for agents to know (or estimate) of their total number.
 As we later show, the lack of such knowledge may have a non-negligible impact on the  performance.

 \paragraph{Our Results:}
 We first show that if the agents have a constant approximation of their total number $k$ then there exists  a rather simple search algorithm 
 whose expected running time is $O(D+D^2/k)$, making it 
 $O(1)$-competitive.
We then turn our attention to  \emph{uniform} searching algorithms, in which agents are not assumed to have any  information regarding $k$.
We prove  that the speed-up penalty for using uniform algorithms is slightly more than logarithmic in the number of agents. Specifically,
we show that, for every constant $\epsilon> 0$,  there exists a uniform search algorithm that is $O(\log^{1+\epsilon} k)$-competitive.
 On the other hand, we show that there is no uniform search algorithm that is $O(\log k)$-competitive. 
In addition, we give a lower bound for the intermediate  setting in which agents are given some estimation of $k$. As a special case, this lower bound implies that for any constant $\epsilon>0$, if each agent  is  given a (one-sided)  $k^\epsilon$-approximation to $k$, then the competitiveness is 
$\Omega(\log k)$.
Informally, our results imply that the agents can potentially perform well without any knowledge of their total number $k$, however, to further improve they must be given a relatively good approximation of $k$.
 Finally, we propose a uniform search algorithm that is concurrently efficient and extremely simple which may imply some  relevance for actual biological scenarios.

\paragraph{Related Work:}
%In biological studies it has been observed that, searching in groups can increase the
%efficiency of locating novel food sources (e.g., see [HW90]). In some extreme cases, food is
%so scarce that group searching is believed to be required for survival [CCG99, JS98].
%Foraging in groups can also facilitate the defense of larger territories [S71]. Overlapping
%searches, especially when restricted to a small area, may decrease individual efficiency
%within a group. It has therefore been suggested that search overlaps must be decreased to a
%minimum, in order to increase the efficiency of the group [M71, VABMPS9]
Collective search is a classical problem that has been extensively studied in different fields of science.  Group living and food sources that have to be actively sought after make collective foraging a widespread biological phenomenon.  Social foraging theory \cite{GC00} makes use of economic and game theory to optimize food exploitation as a function of group size and degree of cooperativity between agents  in different environmental settings. Social foraging theory has been extensively compared to experimental data (see, e.g., \cite{AMFL10,GG88}) but does not typically account for the spatial characteristics of resource abundance. Central place foraging theory \cite{OP79}  assumes a situation in which food is collected from a patchy resource and is returned to a particular location such as a nest. This theory is used to calculate optimal durations for exploiting food patches at different distances from  the central location and has also been tested against experimental observations \cite{GKDL94,HO87}. Collective foraging around a central location is particularly interesting in the case of social insects where large groups forage cooperatively with, practically, no competition between individuals.  Harkness and Mardouras \cite{HM85} have conducted a joint  experimental and modeling research into the collective search behavior of non-communicating desert ants. Modeling the ants' trajectories using biased random walks, they reproduce some of the experimental findings and demonstrate significant speed-up with group size. In bold contrast to these random walks, Reynolds \cite{R06}  argues that L\'evy flights with a power law that approaches unity is the optimal search strategy for cooperative foragers as traveling in straight lines tends to decrease overlaps between searchers.

From an engineering perspective, 
the distributed cooperation of a team of autonomous agents (often referred to as robots or UAVs - Unmanned Aerial Vehicles) is a problem that has been extensively studied. These models extend single agent searches in which an agent  with limited sensing abilities attempts to locate one or several mobile or immobile targets \cite{PYP01}. The memory and computational capacities of the agent are typically large and many algorithms rely on the construction of cognitive maps of the search area that includes current estimates that the target resides in each point \cite{YMP04}. The agent then plans an optimal path within this map with the intent, for example, of optimizing the rate of uncertainty decrease \cite{KBG06}. Cooperative searches typically include communication between the agents that can be transmitted up to a given distance, or even without any restriction.   Models have been suggested where agents can communicate by
altering the environment to which other agent then react \cite{WAYB08}.  Cooperation without
communication has also been explored to some extent \cite{RCA92} but  the analysis puts no emphasis on
the speed-up of the search process. In addition, to the best of our knowledge, no works exist in this context that put emphasis  on finding nearby targets faster than faraway one. Similar problems studied in this context are pattern formation \cite{DFSY10,SY99}, 
%gathering (or 
rendezvous \cite{AP04,FPSW05}, and flocking \cite{GP03}. It is important to
stress, that in all those engineering works, the issue of whether
robots know  their total number is typically not addressed, as obtaining such
information does not seem to be problematic. Furthermore, in many 
works, robots are not identical and have unique identities.

%Little is known about the solvability of other problems like spreading and exploration. Moreover, in the
%area of reliability and fault-tolerance, limited range of visibility, obstacles that limit the
%visibility, as well as robots that appear and disappear from the scene, are all topics that
%have not yet been studied.
In the theory of computer science,
the exploration of graphs using mobile agents is a central question.
Most of the research for graph exploration is concerned with the case
of a single deterministic  agent exploring a finite graph (typically, with some restrictions on the resources
of the agent and/or on the graph structure). For example, in \cite{AH00, BFRSV, DP99, FG05} the agent explores strongly connected
directed finite graphs,  and in \cite{DP02, DFKP02, DKK01, GPRZ07, PP99, R08} the agent explores undirected finite graphs.
 When it comes to probabilistic searching, the random
walk is a natural candidate, as  it is extremely simple, uses no memory, and
trivially self-stabilizes. Unfortunately, however, the random walk turns out to be inefficient
in a two-dimensional infinite grid. Specifically, in this case, the expected hitting time is infinite,
even if the treasure is nearby.

% The more complex case with k-random walkers has recently
%gained attention in the literature [AAKKLT11, ES11]. In particular, for the finite grid, as
%long as k is polynomial in n, the speed-up is only logarithmic in k. The situation with
%infinite grids is even worse. Specifically, though the k-random walks will find the treasure
%with probability one, the expected time to find the treasure will be infinite.
%

%From a theoretical point of view, graph exploration and searching for treasures are very related tasks. These tasks are central questions in  computer science.
%How to explore graphs using mobile agents (sometimes called robots) is a central  question in  computer science.
%Algorithms for graph exploration by deterministic agents  have been intensely studied.

 Evaluating the time to find the treasure as a function of $D$, the initial distance to the treasure, was studied in the context of the cow-path problem.
%is the following problem: a Cow (agents) is standing on the crossroad (refer as the origin) with $m$ infinite paths leading into unknown territory. On one of the paths there is a treasure (Meadow Field) at distance $D$ from the origin. The goal of the searchers is to find the treasure. The cow has no knowledge were the treasure is, or what is the distance $D$. The question is how to find the treasure effectively.
One of the first papers that studied the cow-path problem is the paper  by Baeza-Yates et al. \cite{BCR91}, which shows that the competitive ratio for deterministically finding~a~point on the real line is nine. Considering the two-dimensional case, Baeza-Yates et al.  prove that the spiral search algorithm is optimal up to lower order terms. Randomized algorithms for the problem were studied by  Kao et al. \cite{KRT96},
for the infinite star topology.  Karp et al.  \cite{KSW86} studied an early variant of the cow-path problem on a binary tree. Recently,  Demaine et al. \cite{DFG2006} has considered  the cow-path problem with  a double component price: the first is distance and the second is turn cost. In \cite{LS01}, L\'opez-Ortiz and  Sweet extended the cow-path problem by considering $k$ agents. However, in contrast to our setting, the agents they consider are not identical, and the goal is achieved by (centrally) designing a different specific path for each of the $k$ agents.

In general, the more complex setting on using multiple identical agents has received much less attention.
Exploration by deterministic multiple agents was studied in, e.g.,~\cite{AB96, AB97,FDPS11,FGKP04}.
To obtain better results when using several identical deterministic agents, one must assume that the agents are either centrally coordinated or
that they have some means of  communication (either explicitly, or implicitly, by being able to detect the presence of nearby agents).
%Broadly speaking, the difficulty with deterministic agents comes from symmetry breaking issues.  When it comes to probabilistic agents, such
%difficulties can sometimes be solved quite efficiently. The natural algorithm consisting of an agent performing a \emph{random walk} is a good example.
%In fact, in some respects, random walks are very efficient searching algorithms, at least, for finite graphs. They are extremely  simple, they use no memory,
%and they trivially self-stabilize.
%Unfortunately, however, random walks turn out to be inefficient in an infinite grid. Specifically, if the grid is infinite 2-dimensional then the expected hitting time is infinite.
%
When it comes to probabilistic agents, 
analyzing the speed-up measure for  $k$-random walkers
has recently gained attention.
In a series of papers, initiated by Alon et al.\cite{AAKKLT}, a speed-up of $\Omega(k)$ is established for various finite graph families,
including, in particular, expenders and random graphs~\cite{AAKKLT, ES11, CFR09}.
%Alon et al.
%, developed several bounds that are based on
%the quotient between the cover time and maximum hitting times. Their
%technique
%shows a speed-up of $\Omega(k)$ for many graph families, e.g., for expenders.
%In  \cite{ES11}, Elsasser and  Sauerwald
%present a new lower bound on the speed-up that depends on the
%mixing-time. It gives a speed-up of $\Omega(k)$ on various graphs, even if $k$
%is as large as $n$. In \cite{CFR09}, Cooper et al.,  show a linear speed-up in $k$ for random graphs.
While some graph families enjoy linear speed-up, for many
graph classes, to obtain linear speed-up, $k$ has to be quite small.
In particular, this is true for the two-dimensional $n$-node grid, where a linear speed up is obtained when $k<O(\log^{1-\epsilon} n)$.
On the other hand, the cover time of $2$-dimensional $n$-node grid
is always $\Omega(n/ \log k)$, regardless of $k$. Hence, when $k$ is polynomial in $n$, the speed up is only logarithmic in~$k$.
The situation with infinite grids is even  worse.
Specifically, though the $k$-random walkers would find the treasure with probability one, the expected time to find the treasure becomes infinite.
%In the case were the grid is finite having $n$ nodes it is convenient to measure the performance by speed up that the $k$-walks gives. In this case it was proved in \cite{AAKKLT} that the cover time of $k$-random walks is always bigger than $\frac{n}{\log(k)}$. This is proof is also implies that the hitting time is bigger than $\frac{n}{\log(k)}$.

The question of how important it is for individual processors to know their total number has recently been addressed in the context of locality.  Generally speaking, it has been observed that for several 
classical local computation tasks, knowing the number of processors is not essential~\cite{KSV11}. On the other hand, in the context of   local decision, some evidence exist that such knowledge may be crucial for non-deterministic distributed decision~\cite{FKP11}.

%

%Ð We find the correct order of the speed-up for any value of 1  k  n
%on hypercubes, random graphs and expanders. For d-dimensional
%torus graphs (d > 2), our bounds are tight up to a factor of O(log n).
%Ð Our findings also reveal a surprisingly sharp dichotomy on several
%graphs (including d-dim. torus and hypercubes): up to a certain
%threshold the speed-up is k, while there is no additional speed-up
%above the threshold.

%It is known that two cooperating agents can learn exactly any strongly-connected directed graph with indistinguishable nodes in expected polynomial time~\cite{BS94}.

%Recently there are several works that try to extend the classical searcher problems in the single Turing machine model to multi searchers problems in the context of distributed computing, see for example \cite{AAKKLT}. We believe that the right measure for the performance of the distributed algorithm is the speed-up it gains over the single searcher algorithm. This approach has two basic advantages, it bridges the classical results for one searcher $k=1$ to the general distributed case $k$ searcher problem. Anther advantage of the speed-up approach is that we can neglect low order terms. This happens since we are looking at the ratio of time it takes to find the treasure for a single searcher and the time it takes to find the treasure for $k$-searchers. The low order terms disappear when we look at the asymptotes.

\section{Preliminaries}\label{sec:preliminaries}
%\paragraph{The model:}
We consider the problem where $k$ mobile {\em agents} (robots) are searching for a {\em treasure} on the two-dimensional plane. 
Each agent has a bounded field of view of say $\varepsilon>0$, hence, for simplicity, we can assume that the agents  are actually walking on the integer two-dimensional infinite grid $G=\mathbb{Z}^2$. All $k$ agents starts the search from a central node $s\in G$, called the {\em source}.
An adversary locates the treasure  at some node $\tau\in G$, referred to as the {\em target} node; the agents have no a priori information about the location of $\tau$. The goal of the agents it to {\em find} the treasure-- this task is accomplished once at least one of the agents visits the node~$\tau$. 

The agents are probabilistic machines that can move on the grid, but cannot communicate between themselves. 
All $k$ agents are identical (execute the same protocol).  An agent can traverse an edge of the grid in both directions.
We do not restrict the internal storage  and computational power of agents, nevertheless, we note that all our upper bounds use simple procedures that can be implemented using relatively short memory. 
 For example, with respect to navigation, our constructions only  assume  the ability to perform four basic procedures, namely: (1) choose a 
direction uniformly at random, (2) walk in a "straight" line to a prescribed distance, (3) perform a {\em spiral search}  around a node\footnote{The spiral search around a node $v$ is a particular deterministic local search algorithm (see, e.g., \cite{BCR91}) that starts at $v$ and enables the agent  to visit all nodes at distance  $\Omega(\sqrt{x})$ from $v$ by traversing $x$ edges, for every integer $x$. For our purposes, since we are concerned only with asymptotic results, we can replace this atomic
navigation procedure with any procedure that guarantees such a property. For simplicity, in the remaining of the paper, we assume that for every integer $x$, the spiral search of length $x$ starting at a node $v$ visits all nodes at distance at most $\sqrt{x}/2$ from $v$}, and (4) return to the source node. On the other hand, for our lower bounds to hold, we do not require any restriction on the navigation capabilities.
%The particular way in which agents implement such navigations is beyond the scope of this paper.

Regarding the time complexity, we assume that a traversal of an edge it performed in 1 unit of time. Furthermore, for the simplicity of presentation, we assume that the agents are synchronous, that is, each edge traversal costs precisely 1 unit of time (and all internal computations are performed in zero time). Indeed, this assumption can easily be removed if we measure the time according to the slowest edge-traversal. We also assume that all agents start the search simultaneously at the same time, denoted by $t_0$. This assumption can also be easily removed by starting to count the time after the last agent initiates the search. We measure the cost of an algorithm  by its {\em expected running time},  that is, the expected time (from time $t_0$) until at least one of the agents finds the treasure.
We denote the expected running time of algorithm $\cA$ by $\cT_{\cA}(D,k)$.

The {\em distance} between two nodes  $u,v\in G$, denoted $d(u,v)$, is simply the hop distance between them, i.e., the number of edges on the shortest path connecting $u$ and $v$.  Let $B(r)$ denote the ball centered at the source $s$ with radius $r$, formally, $B(r)=\{v\in G: d(s,v)\leq r\} $.
Denote the distance between the source node $s$ and the target node $\tau$  by $D$, i.e., $D=d(s,\tau)$. Note that if an agent knows $D$, then it can potentially find the treasure in time $O(D)$, by
walking to a distance $D$ in some direction, and then performing a circle around the source of radius $D$ (assuming, of course, that its navigation abilities enable it  to perform such a circle). On the other hand, with the absence of knowledge about $D$, an agent can find the treasure in time $O(D^2)$ by performing a spiral search  around the source. When considering $k$ agents, it is easy to see\footnote{To see why, consider a search algorithm $\cA$ whose expected running time is $T$. Clearly,  $T\geq D$, because it takes $D$ time to merely reach the treasure.
Assume, towards contradiction, that $T<D^2/4k$.  In any execution of $\cA$, by time $2T$, the $k$ agents can visit a total of at most $2Tk<D^2/2$ nodes. Hence,
 by time  $2T$, more than half of the nodes in $B_D:=\{u\mid 1\leq d(u)\leq D\}$
were not visited. Therefore, there must exist a node
$u\in B_D$  such that the probability that $u$ is visited by time $2T$ (by at least one of the agents) is less than $1/2$. If the adversary locates the treasure at $u$ then the expected time to find the treasure is strictly greater than~$T$, which contradicts the assumption.} that the expected running time is $\Omega(D+D^2/k)$, even if the number of agents $k$ is known to all agents, and even if we relax the model and allow agents to freely communicate between each other. It follows from Theorem~\ref{thm:know-k} that if $k$ is known to agents then there exists a search algorithm whose expected running  time is asymptotically optimal, namely, $O(D+D^2/k)$. We evaluate the performance of an algorithm that does not assume the precise knowledge of $k$ with respect to this aforementioned optimal time.
Formally,  let $\phi(k)$ be a function of $k$. An  algorithm $\cA$ is called {\em $\phi(k)$-competitive}
if $$T_{\cA} (D,k)\leq \phi(k)\cdot (D+D^2/k),$$ for every integers $k$ and $D$. We shall be particularly interested in the performances of {\em uniform} algorithms-- these are algorithms in which no information regarding $k$ is available to agents.   (The term uniform is chosen to stress that
agents execute the same algorithm regardless of their number, see, e.g., \cite{KSV11}.)

\section{Upper Bounds}\label{sec:upper}

\subsection{Optimal Running Time with Knowledge on $k$}
Our first theorem asserts that agents can obtain asymptotically optimal running time if they know the precise value of $k$.
As a corollary, it will follow that, in fact, to obtain such a bound, it is sufficient to assume that agents only know a constant approximation of $k$.
Due to lack of space, the proofs of Theorem~\ref{thm:know-k} and Corollary~\ref{cor:known} are deferred to the Appendix. 
\begin{theorem}\label{thm:know-k}
%Fix a constant $\rho\ge1$.
    Assume that the agents know the value of $k$. Then,
    there exists a (non-uniform) search algorithm running in expected time $O(D+D^2/k)$. %Hence, $\cT(D,k)= \Theta(D+D^2/k)$. 
\end{theorem}

Fix a constant $\rho\ge1$. We say that the agents have a $\rho$-approximation of $k$, if, initially, each agent $a$ receives as input a value $k_a$ satisfying 
$k/\rho\le k_a\le k\rho$.

\begin{corollary}\label{cor:known}
Fix a constant $\rho\ge1$.
    Assume that the agents have a $\rho$-approximation of $k$. Then,
    there exists a (non-uniform) search algorithm which is $O(1)$-competitive.
\end{corollary}

\subsection{Unknown Number of Agents}
%
%\subsection{Upper Bounds}
%
We now turn our attention to the case of uniform algorithms.
\begin{theorem}
For every positive constant $\varepsilon$, there exists a uniform search algorithm that is $O(\log^{1+\varepsilon}k)$-competitive.
\end{theorem}
\begin{proof}
Consider the uniform search algorithm $\cA_{\uniform}$ described  below. Let us analyze the performances of the algorithm, and show that its expected running time is $T(D,k):=\phi(k)\cdot(D+D^2/k)$, where $\phi(k)=O(\log^{1+\varepsilon}k)$.
We first note that it suffices to prove the statement
when $k\le D$.
%\log^{\sqrt{1+\varepsilon}}k$.
Indeed, if $k>D$, then we may consider only $D$ agents among
the $k$ agents and obtain an upper bound on the running time of $T(D,D)$,
which is less than $T(D,k)$.
\renewcommand{\baselinestretch}{1.1}
\begin{algorithm}
%\DontPrintSemicolon
\Begin{
Each agent performs the following\;
%$(G_{1},\inp_{1})\longleftarrow(G,\inp)$\;
\For(the big-stage $\ell$){$\ell$ from $0$ to $\infty$}{%
\For(the stage $i$){$i$ from $0$ to $\ell$}{%
\For(the phase $j$){$j$ from $0$ to $i$}{%
 \begin{smallitemize}
 \item
 Set $k_j\longleftarrow2^j$~~and~~$D_{i,j}\longleftarrow  \sqrt{2^{(i+j)}/j^{(1+\varepsilon)}}$\;
%\item 
%$D_{i,j}\longleftarrow  \sqrt{2^{(i+j)}/j^{(1+\varepsilon)}}$ \;
\item 
Go to node $u\in B(D_{i,j})$ chosen uniformly at random among the nodes
in $B(D_{i,j})$\;
\item 
Perform a spiral search starting at $u$ for $t_{i,j}=2^{i+2}/j^{1+\varepsilon}$ time\;
%\footnote{The constant factor is such that guarantees that the
%ball of radius $\sqrt{t_{i,j}}$ is covered by the spiral search. In particular, chossing it to be 4 will}\;
\item 
Return to the source
\end{smallitemize}
  }
 }
}
}
\caption{The uniform algorithm $\cA_{\uniform}$.}\label{alg:u-alg}
\end{algorithm}
\begin{assert}\label{as:1}
For every integer $\ell$, the time until all agents complete big-stage $\ell$ is $O(2^\ell)$.
\end{assert}
For the assertion to hold, it is sufficient to prove that stage $i$
in big-stage $\ell$ takes $O(2^{i})$ time.
To this end, notice that phase $j$ takes $O(D_{i,j}+2^{i}/j^{1+\varepsilon})\leq O(2^{(i+j)/2}+2^{i}/j^{1+\varepsilon}))$ time.
Therefore, stage $i$ takes time
\[
O\left(\sum_{j=0}^{i}(2^{(i+j)/2}+2^{i}/j^{1+\varepsilon})\right)=O(2^{i}).
\]This establishes Assertion~\ref{as:1}.

Let $s=\lceil\log((D^2\cdot\log^{1+\varepsilon} k)/k)\rceil+1$. Fix an integer
$i\ge s$.
Then, there exists $j\in\{0,\ldots,i\}$ such that $2^{j}\le
k<2^{j+1}$.\\
\begin{assert}\label{as:2}
The probability that none of the agents  finds the treasure while executing  phase $j$ of stage
$i$ is at most $c$, for some constant $c<1$.
%at least one of the agents  finds the treasure while executing  phase $j$ of stage
%$i$ is $\Omega(1)$.
\end{assert}
To see this, first note that the treasure is inside the ball $B(D_{i,j})$.
Indeed,
$D_{i,j}=\sqrt{\frac{2^{i+j}}{j^{1+\varepsilon}}}\ge\sqrt{\frac{2^{s+j}}{j^{1+\varepsilon}}}>D$.
Now, observe that the total number of nodes in $B(D_{i,j})$ is $O(D_{i,j}^2)=O(2^{i+j}/j^{1+\varepsilon})$.
Moreover, at least half of the ball of radius $\sqrt{t_{i,j}}$ around the
treasure is contained in $B(D_{i,j})$.
Consequently, the probability for an agent $a$ to choose a node $u$ in
a ball of radius $\sqrt{t_{i,j}}$ around the treasure in phase $j$ of stage
$i$ is
\[
\Omega\left(t_{i,j}/|B(D_{i,j})|\right)=\Omega\left(\frac{2^i/j^{1+\varepsilon}}{2^{i+j}/j^{1+\varepsilon}}\right)=\Omega\left(2^{-j}\right).
\]
If this event happens, then the treasure is found during the corresponding spiral search of agent $a$.
As a result, there exists a positive constant $c'$ such that
the probability that none of the $k$ agents finds the treasure
during phase $j$ of stage $i$ is at most
$(1-c'\cdot 2^{-j})^k\le(1-c'\cdot 2^{-j})^{2^j}\le e^{-c'}$.
%In other words, the probability that one of the $k$ agents finds the treasure during
%phase $j$ of stage $i$ is at least $1-e^{-c}$. 
This establishes Assertion~\ref{as:2}.

By the time that all agents have completed their respective
big-stage $s+\ell$, all agents have performed
$\Omega(\ell^2)$ stages $i$ with $i\ge s$. By Assertion~\ref{as:2},
for each such $i$, the probability that the treasure is not found
during stage  $i$ is at most $c$ for some constant $c<1$.
Hence, the probability that the treasure is not found
during any of those $\Omega(\ell^2)$ stages is at most $1/d^{\ell^2}$ for some
constant $d>1$.
Assertion~\ref{as:1} ensures that all agents complete big-stage $s+\ell$ by time $O(2^{s+\ell})$,
so the expected running time is
$O(\sum_{\ell=0}^\infty 2^{s+\ell}/d^{\ell^2})=O(2^s)=O(D^2\log^{1+\varepsilon} k/k)$, as desired.
\end{proof}

\section{Lower Bounds}
\subsection{An Almost Tight Lower Bound for Uniform Algorithms}
%We first show that our upper bound for the competitiveness of the uniform algorithm $\cA_{\uniform}$ is almost tight.

\begin{theorem}\label{th:lowerbound}
There is no uniform search algorithm that is $O(\log k)$-competitive.
\end{theorem}
\begin{proof}
Suppose that there exists a uniform search algorithm with running time less
than $f(D,k)=(D+D^2/k)\Phi'(k)$. Hence, as long as $k\leq D$, we have
$f(D,k)\leq \frac{D^2\Phi(k)}{k}$, where $\phi(k)=2\phi'(k)$. Assume
towards a contradiction that $\phi(k)=O(\log k)$.

Let $T$ be a (sufficiently large) integer, and let $D=2T+1$.
That is, for the purpose of the proof, we assume that the treasure is actually placed at some far away distance $D=2T+1$. This means, in particular,
that by time $2T$ the treasure has not been found yet.

For every integer $i\leq \log T/2$, set
\[
k_i=2^i\quad\text{and}\quad D_i=\sqrt{T\cdot k_i/ \phi(k_i)}.
\]

Fix an integer $i$ in $[1,\log T/2]$, and consider  $B(D_i)$, the ball of radius $D_i$ around the source node.
We consider now the case where the algorithm is executed with $k_i$ agents.
For every set $S\subseteq B(D_i)$, let $\chi(S)$ be the random variable
indicating the number of nodes in $S$ that were visited by at least one of the
$k_i$ agents by time $2T$. (For a singleton node $u$, we write $\chi(u)$ for $\chi(\{u\})$.)

Note that $k_i\leq D_i$, and therefore, for each node $u\in B(D_i)$, the expected time to visit $u$ is at most  $f(D_i,k_i)\leq{D_i^2\phi(k_i)}/{k_i}= T$.
Thus, by Markov's inequality, the probability that $u$ is visited by time $2T$ is at least $1/2$,
i.e.,~$\mathbf{Pr}(\chi(u)= 1)\geq 1/2$. Hence, $\mathbf{E}(\chi(u))\geq 1/2$.

Now consider an integer $i$ in $[2,\log T/2]$, and set $S_i=B(D_i)\setminus B(D_{i-1})$.
By linearity of expectation, $\mathbf{E}(\chi(S_i))=\sum_{u\in S_i} \mathbf{E}(\chi(u))\geq |S_i|/2$.
Consequently, by time $2T$, the expected number of nodes in $S_i$ that an agent visits is
\[
\Omega(|S_i|/k_i)=\Omega\left(\frac{D_{i-1}(D_i-D_{i-1})}{k_i}\right)=\Omega\left(\frac{T}{\phi(k_{i-1})}\cdot \left(  \frac{\sqrt{2\phi(k_{i-1})}}{\sqrt{\phi(k_{i})}} - 1 \right)\right)=\Omega \left(\frac{T}{\phi(k_i)}\right),
\]
where the second equality follows from the fact that $D_i=D_{i-1}\cdot
\sqrt{\frac{2\phi(k_{i-1})}{\phi(k_{i})}}$, and the third equality follows from the fact that $\phi(k_i)=O(i)$.
In other words, for every integer $i$ in $[2,\log T/2]$, the expected number of
nodes in $S_i$ that each agent visits by time $2T$ is $\Omega\left(\frac{T}{\phi(k_i)}\right)$.
Since the sets $S_i$ are pairwise disjoint, the linearity of expectation
implies that the expected number of nodes that an agent visits by time $2T$ is
\[
\Omega\left(\sum_{i=2}^{\log
T/2}\frac{T}{\phi(k_i)} \right)=T\cdot \Omega\left(\sum_{i=2}^{\log
T/2}\frac{1}{\phi(2^i)} \right).
\]
Hence, the sum $\sum_{i=2}^{\log T/2}\frac{1}{\phi(2^i)}$ must converge as $T$ goes to infinity.
This contradicts the assumption that $\phi(k)=O(\log k)$.
\end{proof}

\subsection{A Lower Bound for Algorithms using Approximated Knowledge of $k$}
We now present a lower bound for the competitiveness of search algorithms assuming that agents have approximations for $k$. 
As a special case, our lower bound implies that for any constant $\epsilon>0$, if agents are given an estimation $\tilde{k}$ such that 
$\tilde{k}^{1-\epsilon}\leq k\leq \tilde{k}$, then the competitiveness is $\Omega(\log k)$. That is, the competitiveness remains logarithmic even for relatively good approximations of $k$.

%The following theorem implies that if the agents wish to pay very little penalty then they must have some information about $k$.
Formally, let $\epsilon(x)$ be some function such that $0<\epsilon(x)\leq 1$, for every integer $x$. We say that the agents have $k^\epsilon$-approximation of $k$ if each agent $a$ receives as input an estimation $\tilde{k}_a$ for $k$, satisfying: 
$$\tilde{k}_a^{1-\epsilon(\tilde{k}_a)}\leq k\leq \tilde{k}_a.$$
(For example, if $\epsilon(x)$ is the constant $1/2$ function and if the  agents  have $k^\epsilon$-approximation of $k$, then this means, in particular,  
that if all agents receive the same value $\tilde{k}$ then  the real number of agents  $k$ satisfies  $\sqrt{\tilde{k}}\leq k\leq \tilde{k}.$)
%Supposed that all agents know an estimate $\tilde{k}$ of $k$, such that $\tilde{k}^{1-\epsilon(\tilde{k})}\leq k\leq \tilde{k}$. 

\begin{theorem}\label{th:lowerbound}
Let $\epsilon(x)$ be some function such that $0<\epsilon(x)\leq 1$, for every integer $x$. 
Consider the case that the agents have a $k^\epsilon$-approximation of $k$. 
Suppose that there exists  a $\phi(k)$-competitive algorithm, where $\phi$ is non-decreasing. Then, $\phi(k)=\Omega(\epsilon(k)\log k)$.
%Then there is no search algorithm that is $o(\log k)$-competitive.
\end{theorem}
\begin{proof}
Assume  that there is a search algorithm
for this case running in time $(D+D^2/k)\phi(k)$, where  $\phi$ is non-decreasing. 
Suppose that all agents receive the same  value $\tilde{k}$, that should serve as an estimate for $k$. 
%Moreover, suppose that the function 
%$\epsilon(\cdot)$ is known to the agents.
Consider a large enough integer $W$, specifically,  such that $4\tilde{k}< W$. 
Set $$T=2W\cdot\phi(\tilde{k})    \mbox{~~~and~~~} j_0=\frac{\log W}{2}.$$ 
For the purposes of the proof, we assume that the treasure is located at distance $D=2T+1$, so that by time $2T+1$ it is guaranteed that no agent finds the treasure. 
%(Recall, the agents are not aware of  $D$.)

 For $i=1,2,\cdots$, define 
 $$
 S_i:=\{u\mid 2^{j_0+i-1}<d(u,s)\leq 2^{j_0+i}\}.
 $$
 Fix an integer $i\in\{\lceil \frac{1-\epsilon(\tilde{k})}{2}\log \tilde{k}\rceil, \cdots,\lfloor\frac{1}{2}\log \tilde{k}\rfloor\}.$ Assume for the time being, that the number of agents is $k_i=2^{2i}$. Note that $$ \tilde{k}^{1-\epsilon(\tilde{k})}\leq k_i \leq \tilde{k}, $$
 hence, $k_i$ is a possible candidate for being the real number of agents.  Observe that all nodes in $S_i$ are at distance at most $2^{j_0+i}$ from the source, and that $|S_i|=\Theta(2^{2j_0+2i})=\Theta(W\cdot k_i)$. 
By definition, it follows that  $j_0\geq i+1$. Hence,
  $k_i\leq 2^{j_0+i-1} < d(u,s)$, and therefore 
 it follows   by the required expected time of the algorithm, that for each node $u\in S_i$, the expected time to cover by at least one of the $k_i$  agents  is at most 
 $$
 \frac{2d(u,s)^2}{k_i} \cdot\phi(k_i)~~\leq~~ 2W\cdot\phi(\tilde{k})=T.
 $$
 
Recall that we now  consider the case where the algorithm is executed with $k_i$ agents.
For every set of nodes $S\subseteq G$, let $\chi(S)$ be the random variable
indicating the number of nodes in $S$ that were visited by at least one of the
$k_i$ agents by time $2T$. (For a singleton node $u$, we write $\chi(u)$ for $\chi(\{u\})$.)
By Markov's inequality, the probability that $u$ is visited by at least one of the $k_i$ agents by time $2T$ is at least $1/2$,
i.e.,~$\mathbf{Pr}(\chi(u)= 1)\geq 1/2$. Hence, $\mathbf{E}(\chi(u))\geq 1/2$.
 By linearity of expectation, $\mathbf{E}(\chi(S_i))=\sum_{u\in S_i} \mathbf{E}(\chi(u))\geq |S_i|/2$.
Consequently, by time $2T$, the expected number of nodes in $S_i$ that a single agent visits is $\Omega(|S_i|/k_i)= \Omega(W)$.

 Since this holds for any $i\in\{\lceil \frac{1-\epsilon(\tilde{k})}{2}\log \tilde{k}\rceil, \cdots,\lfloor\frac{1}{2}\log \tilde{k}\rfloor\}$, and 
since the $S_i$'s are pairwise disjoint sets, it follows by linearity of expectation, that the expected number of nodes  that a single agent visits
by time $2T$ is  $\Omega(W\cdot \epsilon(\tilde{k})\log \tilde{k})$. Since $T=2W\cdot\phi(\tilde{k})$, this implies that 
 $\phi(\tilde{k})=\Omega(\epsilon(\tilde{k}) \log \tilde{k})$, as desired. This concludes the proof of the theorem.
\end{proof}

\section{The Harmonic Search Algorithm}
The algorithms described in the Section~\ref{sec:upper} are relatively simple but still require the use of
non trivial iterations, which may be complex for simple and tiny agents such as ants.
If we relax the requirement of bounding the expected running time and demand only that the treasure
be found with some low constant probability, then it is possible to avoid one
of the loops of the algorithms. However, a sequence of iterations still needs
to be performed.

In this section, we propose an extremely simple algorithm, coined
the \emph{harmonic search algorithm}\footnote{The name harmonic was chosen
because of structure resemblances to the celebrated harmonic algorithm for the
$k$-server problem --- see, e.g.,~\cite{BG91}.}, which does not perform in iterations
and is essentially composed of three components:
(1) choose a random direction and walk in this direction for a distance $d$,
chosen randomly according to a distribution in which the probability of
choosing $d$  is roughly inverse proportional to $d$,
(2) perform a local search (e.g., a spiral search) for roughly $d^2$ time, and
(3) return to the source.
It turns out that this extremely simple algorithm has a good probability of quickly finding the treasure,
if the number of agents is sufficiently large.

More specifically, the algorithm depends on a positive constant parameter $\delta$ that is fixed in
advance and governs the performance of the algorithm.
For a node $u$, let $p(u)=\frac{c}{d(u)^{2+ \delta}}$, where $c$ is
the normalizing factor, defined so that
$\sum_{u\in V(G)} p(u)=1$. (Note that $c$ depends on $\delta$.)

\renewcommand{\baselinestretch}{1.2}
\begin{algorithm}
%\DontPrintSemicolon
\Begin{
Each agent performs the following three actions\;
1. Go to a node $u\in V(G)$ with probability $p(u)$\;
2. Perform a spiral search for $t(u)=d(u)^{2+\delta}$ time\;
3. Return to the source
}
\caption{The harmonic search algorithm.}\label{alg:h-alg}
\end{algorithm}

\begin{theorem}
Let $\delta\in(0,0.8]$. For every $\varepsilon>0$,
there exists a positive real number $\alpha$ such that
if $k>\alpha D^{\delta}$, then with probability at least $1-\epsilon$, the expected running time of  the harmonic algorithm is $O(D+\frac{D^{2+\delta}}{k}).$
\end{theorem}

\begin{proof}
Fix a real number $\beta$ greater than $\ln(1/\varepsilon)$, so
$e^{-\beta}<\epsilon$. Set $\alpha=12\beta/c$.
We assume that the number of agents $k$ is greater than $\alpha D^{\delta}$ and we show that
with probability at least $1-\epsilon$, the running time of
the harmonic search algorithm is $O(D+\frac{D^{2+\delta}}{k}).$

Let $\lambda = \frac{4\beta D^{1+\delta}}{c k}$. In particular,
$\lambda<D/4$ since $k>\alpha D^{\delta}$.
Consider the ball $B_{\lambda}$ of radius $\sqrt{\lambda D}/2$ around the treasure.
Note that $\sqrt{\lambda D}<D/2$, and hence $3D/4<d(u)<5D/4$ for every node $u\in B_{\lambda}$.
Note also that if $u\in B_{\lambda}$, then an agent that performs a spiral search from $u$
finds the treasure by time $\lambda D$, which is at most $t(u)$ since
$d(u)>3D/4$ and $D\ge2$.

In other words, if an agent goes to a node $u\in {B}_{\lambda}$ in
Step 1 of the algorithm, then this agent finds the treasure in step 2.
Since each node in ${B}_{\lambda}$ is at distance less than $5D/4$ from the
source, it follows that the total running time of the algorithms
is $O(D(5/4+\lambda))$, which is $O(D+\frac{D^{2+\delta}}{k})$.
Hence, let us analyze the probability that at least one of the $k$ agents
goes to a node in ${B}_{\lambda}$ in Step 1 of the algorithm.

Since $d(u)<5D/4$ for each node $u\in {B}_{\lambda}$, the probability $p(u)$ that a single agent
goes to $u$ in Step 1 is at least $\frac{c}{(5D/4)^{2+ \delta}}\geq \frac{c}{2D^{2+ \delta}}$
as $\delta\le0.8$.
Since there are at least $\lambda D/2$ nodes in ${B}_{\lambda}$, the probability that a single agent
goes to a node in ${B}_{\lambda}$ in Step 1 is at least
\[
\sum_{u\in B_{\lambda}}p(u)\ge|B_\lambda|\cdot\frac{c}{2D^{2+ \delta}}\ge
%\frac{c\lambda D}{3D^{2+\delta}}=
\frac{c\lambda}{4D^{1+\delta}}.
\]
%\[
%\frac{c\pi\lambda D}{\pi D^{2+\delta}}= \frac{c\lambda}{D^{1+\delta}}.
%\]
It follows that the probability that no agent goes to a node in $B_{\lambda}$ in Step 1 of the algorithm
is at most
\[
\left(1-\frac{c\lambda}{4D^{1+\delta}}\right)^k=\left(1-\frac{c\lambda}{4D^{1+\delta}}\right)^{\beta\frac{4D^{1+\delta}}{c\lambda}}
\le e^{-\beta}<\varepsilon.
\]
The theorem follows.
\end{proof}

\section{Conclusion and Discussion}
We first presented an algorithm that assumes that agents have a constant approximation of $k$ and runs in optimal $O(D+D^2/k)$ expected time.
We then showed that there exists a uniform search algorithm whose competitiveness  is slightly more than logarithmic, specifically, $O(\log^{1+\epsilon} k)$, for arbitrary small constant $\epsilon>0$.
We also presented a relatively efficient uniform algorithm, namely, the harmonic algorithm, that has extremely simple structure. Our constructions imply that, in the absence of any communication, multiple searchers can still potentially perform rather well.  On the other hand, our lower bound results imply that to achieve better running time,  the searchers must either communicate or utilize some  information regarding $k$. In particular, even if each agent is given a $k^\epsilon$-approximation to $k$ (for constant $\epsilon>0$), it would 
not suffice for being strictly below $O(\log k)$-competitive.

Although the issue of memory
is beyond the scope of this paper, our constructions are simple and can be implemented using relatively low memory. For example, going in a straight line for a distance of $d=2^\ell$ can be implemented using $O(\log\log d)$ memory bits, by employing a randomized counting technique. In addition, our lower bounds result gives evidence that in order to achieve a near-optimal  running time, agents must use non-trivial memory size, required merely to store the necessary approximation of $k$. This may be useful for obtaining a tradeoff between the running time and  the memory size of agents.
%, a task that we believe to be challenging.  

From another perspective, it is of course interesting to experimentally verify whether social insects engage in search patterns in the plane which resemble the simple
uniform algorithms specified above, and, in particular, the harmonic algorithm. Two natural candidates are desert ants {\em Cataglyphys} and honeybees {\em Apis mellifera}.
First, these species seem to face settings which are similar to the one we use.
Indeed, they  cannot rely on communication during the search due to the dispersedness of individuals \cite{HM85} and their inability to leave chemical trails (this is due to increased pheromone evaporation in the case of the desert ant). Additionally, the task of finding the treasure is relevant, as food sources in many cases are indeed relatively rare or patchy. Moreover, due to the reasons mentioned in Section~\ref{sec:introduction}, finding nearby sources of food is of great importance. Second, insects of these species have the behavioral and computational capacity to maintain a compass-directed vector flight~\cite{CSOFF00,HM85}, measure distance using an internal odometer~\cite{SW04,SZAT00}, travel to distances taken from a random power law distribution \cite{RSR}, and engage in spiral or quasi-spiral movement patterns \cite{R08,RSMG07, WRSM81}. These are the main ingredients that are needed to perform the algorithms described in this paper. Finally,   the search trajectories of desert ants have been shown to  include two distinguishable sections:  a long straight path in a given direction emanating from the nest and a second more tortuous path within a small confined area~\cite{HM85,WMZ04}.

\clearpage

\pagenumbering{roman}
\appendix

\renewcommand{\theequation}{A-\arabic{equation}}
\setcounter{equation}{0}
\begin{center}
\textbf{\large{APPENDIX}}
\end{center}

\section{Proof of Theorem~\ref{thm:know-k}.}
%\begin{proof}
% Let $\alpha$ be a
%sufficiently large fixed constant and 
For an integer $i$, let $B_i:= \{u\,:\, d(u) \leq 2^i\}$. Consider the following algorithm.
\renewcommand{\baselinestretch}{1.2}
\begin{algorithm}
%\DontPrintSemicolon
\Begin{
Each agent performs the following double loop\;
\For(the \emph{stage} $j$ defined as follows){$j$ from $1$ to $\infty$}{%
  \For(the \emph{phase} $i$ defined as follows){$i$ from $1$ to $j$}{%
  \begin{smallitemize}
      \item
  Go to a node $u\in B_i$ chosen uniformly at random among the nodes in $B_i$
\item
Perform a spiral search for time $t_i=2^{2i+2}/ k$.
\item
Return to the source $s$
  \end{smallitemize}
  }
 }
}
\caption{The non-uniform algorithm $\cA_{\kk}$.}\label{alg:k-alg}
\end{algorithm}

Fix a positive integer $\ell$ and consider the time $T_\ell$ until each agent completed $\ell$
phases $i$ with $i\geq \log D$.
%Provided that $\alpha$ is sufficiently large,
Each time an agent performs phase $i$, the agent finds the
treasure if the chosen node $u$ belongs to
the ball $B(\tau, \sqrt{t_i}/2)$ around $\tau$, the node holding the treasure.
Note that at least some constant fraction of the ball $B(\tau, \sqrt{t_i}/2)$ is contained in $B_i$.
The probability of choosing a node $u$ in that fraction is thus
$\Omega(|B(\tau,\sqrt{t_i}/2)|/|B_i|)$, which is at least $\beta/k$ for
some positive constant $\beta$.
Thus, the probability that by time $T_\ell$ none of the $k$ agents finds the treasure (while
executing their respective $\ell$ phases $i$) is at most
$(1-\beta/k)^{k\ell}$, which is at most $\gamma^{-\ell}$ for some
constant $\gamma$ greater than $1$.

For an integer $i$, let $\psi(i)$ be the time required to execute a phase $i$.
Note that $\psi(i)=O(2^i+2^{2i}/k)$. Hence, the time until all agents
complete stage $j$ for the first time is
\[\sum_{i=1}^j \psi(i)=O(2^j+\sum_{i=1}^j 2^{2i}/k)=O(2^j+2^{2j}/k).\]
Now fix $s=\lceil\log D\rceil$. It follows that for any integer $\ell$, all agents
complete their respective stages $s+\ell$ by time
$\hat{T}(\ell):=O(2^{s+\ell}+2^{2(s+\ell)}/k)$. Observe that by this time, all agents
have completed at least $\ell^2/2$ phases $i$ with $i\geq s$.
Consequently, the probability that none of the $k$ agents finds the treasure by time
$\hat{T}(\ell)$  is at most
$\gamma^{-\ell^2/2}$. Hence, the expected running time is at most
\[
\cT_{\cA_{\kk}}(D,k)=O\left(\sum_{\ell=1}^\infty
\frac{2^{s+\ell}}{\gamma^{\ell^2/2}}+\frac{2^{2(s+\ell)}}{k\gamma^{\ell^2/2}}\right)=O\left(2^{s}+2^{2s}/k\right)=O(D+D^2/k).\]
This establishes the theorem.
%\end{proof}
\qed
\section{Proof of Corollary~\ref{cor:known}}
    Each agent $a$ executes Algorithm $\cA_{\kk}$ (see the proof of Theorem~\ref{thm:know-k}) with the parameter
    $k$ equal to $k_a/\rho$. Since $k/\rho^2\le k_a/\rho\le k$, the only
    difference between this case and the case where the agents know $k$,
    is that for each agent, the time required to perform each spiral search
    is multiplied by a constant factor of at most $\rho^2$.
    %In other words, the
    %constant $\alpha$ appearing in Algorithm~\ref{alg:k-alg} is increased
    %by at most a multiplicative factor of $\rho^2$. 
    Therefore, the analysis in the proof of Theorem~\ref{thm:know-k}
    remains the same and the running time is increased by a multiplicative factor of at most
    $\rho^2$.
\qed

\end{document}